%% file: main.tex
\documentclass[a4paper,UKenglish,cleveref, autoref, thm-restate]{lipics-v2021}
\usepackage[utf8]{inputenc}
\usepackage[noline,noend,linesnumbered,ruled]{algorithm2e}
\DontPrintSemicolon
\usepackage{tikz}
\usetikzlibrary{calc,decorations.pathmorphing,shapes,shapes.geometric,matrix,positioning}

\usepackage{xcolor}

\nolinenumbers
\hideLIPIcs

\newcommand{\Squares}{\textsf{Squares}}

\renewcommand{\alph}{\textsf{Alph}}

\usepackage{hyperref}

\newcommand{\defdsproblem}[3]{
  \vspace{2mm}
  \noindent
  \textbf{#1}\\
  {\bf{Input:}} #2  \\
  {\bf{Operation:}} #3
  \vspace{2mm}
}

\newcommand{\E}{\mathcal{E}}
\newcommand{\EE}{\mathbf{E}}

\title{Fast Computation of \texorpdfstring{$k$}{k}-Runs, Parameterized Squares, and Other Generalised Squares}

\ccsdesc[500]{Theory of computation~Pattern matching}

\authorrunning{Y.\ Nakashima, J.\ Radoszewski, and T.\ Waleń}

\keywords{string algorithm; $k$-mismatch square; parameterized square; order-preserving square; maximum gapped repeat}

\def\dd{\mathinner{.\,.}}
\newcommand{\floor}[1]{\left\lfloor #1 \right\rfloor}
\newcommand{\ceil}[1]{\left\lceil #1 \right\rceil}
\newcommand{\Oh}{\mathcal{O}}

\newcommand{\LCP}{\textsf{LCP}\xspace}

\newcommand{\LPF}{\textsf{LPF}\xspace}

\newcommand{\TR}{\overrightarrow{T}}
\newcommand{\TL}{\overleftarrow{T}}
\renewcommand{\P}{\mathcal{P}}
\newcommand{\PR}{\overrightarrow{\P}}
\newcommand{\PL}{\overleftarrow{\P}}
\newcommand{\R}{\mathcal{R}}
\renewcommand{\output}{\mathsf{output}}

\author{Yuto Nakashima}{Department of Informatics, Kyushu University, Fukuoka, Japan}{nakashima.yuto.003@m.kyushu-u.ac.jp}{https://orcid.org/0000-0001-6269-9353}{JSPS KAKENHI Grant Number JP25K00136.}
\author{Jakub Radoszewski}{Institute of Informatics, University of Warsaw, Poland}{jrad@mimuw.edu.pl}{https://orcid.org/0000-0002-0067-6401}{Supported by the Polish National Science Center, grant no.\ 2022/46/E/{\allowbreak}ST6/\allowbreak00463.}
\author{Tomasz Wale\'n}{Institute of Informatics, University of Warsaw, Poland}{walen@mimuw.edu.pl}{https://orcid.org/0000-0002-7369-3309}{}

\EventEditors{Anne Benoit, Haim Kaplan, Sebastian Wild, and Grzegorz Herman}
\EventNoEds{4}
\EventLongTitle{33rd Annual European Symposium on Algorithms (ESA 2025)}
\EventShortTitle{ESA 2025}
\EventAcronym{ESA}
\EventYear{2025}
\EventDate{September 15--17, 2025}
\EventLocation{Warsaw, Poland}
\EventLogo{}
\SeriesVolume{351}
\ArticleNo{6}

\Copyright{Yuto Nakashima, Jakub Radoszewski, and Tomasz Waleń}

\begin{document}

\maketitle

\begin{abstract}
A $k$-mismatch square is a string of the form $XY$ where $X$ and $Y$ are two equal-length strings that have at most $k$ mismatches.
Kolpakov and Kucherov [\emph{Theor.\ Comput.\ Sci.}, 2003] defined two notions of $k$-mismatch repeats, called $k$-repetitions and $k$-runs, each representing a sequence of consecutive $k$-mismatch squares of equal length.
They proposed algorithms for computing $k$-repetitions and $k$-runs working in $\Oh(nk\log k+\output)$ time for a string of length $n$ over an integer alphabet, where $\output$ is the number of the reported repeats. We show that $\output=\Oh(nk \log k)$, both in case of $k$-repetitions and $k$-runs, which implies that the complexity of their algorithms is actually $\Oh(nk \log k)$.
We apply this result to computing parameterized squares.

A parameterized square is a string of the form $XY$ such that $X$ and $Y$ parameterized-match, i.e., there exists a bijection $f$ on the alphabet such that $f(X)=Y$.
Two parameterized squares $XY$ and $X'Y'$ are equivalent if they parameterized match.
Recently Hamai et al.\ [SPIRE 2024] showed that a string of length $n$ over an alphabet of size $\sigma$ contains less than $n\sigma$ non-equivalent parameterized squares, improving an earlier bound by Kociumaka et al.\ [\emph{Theor.\ Comput.\ Sci.}, 2016]. We apply our bound for $k$-mismatch repeats to propose an algorithm that reports all non-equivalent parameterized squares in $\Oh(n\sigma \log \sigma)$ time. We also show that the number of non-equivalent parameterized squares can be computed in $\Oh(n \log n)$ time. This last algorithm applies to squares under any substring compatible equivalence relation and also to counting squares that are distinct as strings. In particular, this improves upon the $\Oh(n\sigma)$-time algorithm of Gawrychowski et al.\ [CPM 2023] for counting order-preserving squares that are distinct as strings if $\sigma=\omega(\log n)$.
\end{abstract}

\section{Introduction}
A string of the form $XX$, for any string $X$, is called a \emph{square} (or a \emph{tandem repeat}). 
Squares are a classic notion in combinatorics on words (see, e.g., the early works of Thue~\cite{Thue} on avoidability of square substrings), text algorithms (starting from an algorithm for computing square substrings by Main and Lorentz~\cite{DBLP:journals/jal/MainL84}), and bioinformatics (see Gusfield's book~\cite{DBLP:books/cu/Gusfield1997}).
A string of length $n$ contains at most $n$ distinct square substrings~\cite{brlek2022numbersquaresfiniteword} and all of them can be computed in $\Oh(n)$ time~\cite{DBLP:conf/cpm/BannaiIK17,DBLP:conf/spire/Charalampopoulos20,DBLP:journals/tcs/CrochemoreIKRRW14,DBLP:journals/jcss/GusfieldS04} or even $\Oh(n/\log_\sigma n)$ time~\cite{MFCS25} assuming that the string is over an integer alphabet of size $\sigma$. A maximal sequence of consecutive squares in a string is called a run (or a generalised run; cf.\ \cref{sec:KK++}); see~\cite{DBLP:conf/focs/KolpakovK99}. A string of length $n$ contains at most $n$ runs~\cite{DBLP:journals/siamcomp/BannaiIINTT17} and they can all be computed in $\Oh(n)$ time~\cite{DBLP:journals/siamcomp/BannaiIINTT17,DBLP:conf/icalp/Ellert021}. Our work is devoted to efficient algorithms for computing known generalizations of squares and runs: $k$-mismatch squares represented as $k$-runs or $k$-repetitions, parameterized squares, and generalised squares that include parameterized squares, order-preserving squares, and Cartesian-tree squares.

We assume that positions in a string $X$ are numbered from 1 to $|X|$, so that $X[i]$ is the $i$th character of $X$. For integers $i,j$ such that $1 \le i \le j \le |X|$, by $X[i \dd j]=X[i \dd j+1)$ we denote the substring composed of characters $X[i],X[i+1],\ldots,X[j]$. We use similar notation for integer intervals: $[i \dd j]=[i \dd j+1)=\{i,i+1,\ldots,j\}$.

We say that a length-$n$ string is over an integer alphabet if its letters belong to $[0 \dd n^{\Oh(1)}]$.
We use the word-RAM model of computation.

\subsection{\texorpdfstring{$k$}{k}-Mismatch Squares and \texorpdfstring{$k$}{k}-Runs}
For two equal-length strings $X$ and $Y$, their Hamming distance is defined as $d_H(X,Y)=|\{i \in [1 \dd |X|]\,:\, X[i] \ne Y[i]\}|$. A \emph{$k$-mismatch square} (also known under the name of $k$-mismatch tandem repeat) is a string $XY$ such that $|X|=|Y|$ and $d_H(X,Y) \le k$.
Let $T$ be a string of length $n$.
A \emph{$k$-run of period $\ell$} in $T$ (cf.~\cite{Kucherov2014}) is a maximal substring $T[a \dd b)$ such that $T[i \dd i+2\ell)$ is a $k$-mismatch square for every $i \in [a \dd b-2\ell]$; see \cref{fig:kruns}. Maximality means that the $k$-run is extended to the right and left as much as possible provided that the definition is still satisfied.

\begin{figure}[htpb]
\centering
\include{_fig_kruns}
\caption{String $T[1 \dd 26]$ contains two 2-runs with period 8, $T[1\dd 18]$ and $T[5 \dd 24]$. The two 2-runs represent three 2-mismatch squares (below) and five 2-mismatch squares (above), respectively; mismatches are shown in red. Strings $T[x \dd x+16)$ for $x\in \{4,10,11\}$ are not 2-mismatch squares.}
\label{fig:kruns}
\end{figure}

Landau, Schmidt and Sokol~\cite{DBLP:journals/jcb/LandauSS01} showed an algorithm for computing $k$-runs in a string of length $n$ that works in $\Oh(nk\log(n/k))$ time; hence, the total number of $k$-runs reported is $\Oh(nk\log(n/k))$. Kolpakov and Kucherov~\cite{DBLP:conf/esa/KolpakovK01,DBLP:journals/tcs/KolpakovK03} showed that all $k$-runs (called there runs of $k$-mismatch tandem repeats) in a string of length $n$ can be computed in $\Oh(nk\log k+\output)$ time, where $\output$ is the number of $k$-runs. 
Many other algorithms for computing approximate tandem repeats under various metrics, in the context of computational biology and with the aid of statistical methods and heuristics, were proposed~\cite{DBLP:conf/recomb/Benson98,Benson,ex,doi:10.1089/cmb.2007.0018,DBLP:conf/swat/KaplanPS06,mreps,10.1093/bioinformatics/bth311,e1,10.1093/bioinformatics/btq209,Reneker,404256,DBLP:journals/bioinformatics/RivalsDDDDHO97,DBLP:journals/jcb/SagotM98,Sobreira,10.1093/bioinformatics/btl309,SOKOL2014103,10.1145/974614.974644,4472917}.
In \cref{sec:KK++} we show the following theorem.

\begin{theorem}\label{thm:KK++}
A string of length $n$ contains $\Oh(nk\log k)$ $k$-runs.
\end{theorem}

\cref{thm:KK++} provides the first $\Oh(n)$ upper bound on the number of $k$-runs for $k=\Oh(1)$ and implies that the algorithm of Kolpakov and Kucherov computing $k$-runs actually works in $\Oh(nk\log k)$ time. Actually, we show a stronger condition that a string $T$ of length $n$ contains $\Oh(nk\log k)$ \emph{uniform} $k$-runs. Intuitively, a uniform $k$-run is a maximal sequence of consecutive $k$-mismatch squares of the same length in which the mismatches of all squares are at the same positions. (See~\cref{sec:KK++} for a formal definition.) Our bound on the number of uniform $k$-runs implies the bound for $k$-runs as well as an $\Oh(nk\log k)$ upper bound on the number of $k$-repetitions as defined in~\cite{DBLP:conf/esa/KolpakovK01,DBLP:journals/tcs/KolpakovK03} (called there $k$-mismatch globally defined repetitions).

To prove \cref{thm:KK++}, we explore a combinatorial relation between $k$-runs and maximum gapped repeats~\cite{DBLP:journals/jda/KolpakovPPK17} and apply the optimal $\Oh(n\alpha)$ bound on the number of maximal $\alpha$-gapped repeats in a length-$n$ string~\cite{DBLP:journals/mst/GawrychowskiIIK18,DBLP:journals/tcs/IK19}.

\subsection{Parameterized Squares}
For a string $X$, by $\alph(X)$ we denote the set of characters of $X$.
Two strings $X$, $Y$ \emph{parameterized match} if $|X|=|Y|$ and there is a bijection $f : \alph(X) \mapsto \alph(Y)$ such that $f(X)=Y$ (i.e., $|X|=|Y|$ and $f(X[i])=Y[i]$ for all $i \in [1 \dd |X|]$). A \emph{parameterized square} (\emph{p-square}, in short) is a string $XY$ such that $X$ parameterized matches $Y$ (see~\cref{fig:3gensq}). Two p-squares $XY$ and $X'Y'$ are called \emph{equivalent} if they parameterized match.

Parameterized matching was introduced by Baker~\cite{DBLP:journals/jcss/Baker96} motivated by applications in code refactoring and plagiarism detection. The notion of p-squares was introduced by Kociumaka, Radoszewski, Rytter, and Waleń~\cite{DBLP:journals/tcs/KociumakaRRW16} who showed that a string of length $n$ over alphabet of size $\sigma$ contains at most $2\sigma!n$ non-equivalent p-squares. They also considered avoidability of parameterized cubes. The bounds from~\cite{DBLP:journals/tcs/KociumakaRRW16} were recently improved by Hamai, Taketsugu, Nakashima, Inenaga, and Bannai~\cite{DBLP:conf/spire/HamaiTNIB24}; they showed that a length-$n$ string over alphabet of size $\sigma$ contains less than $\sigma n$ non-equivalent p-squares. This bound automatically implies that such a string contains  at most $\sigma! \cdot \sigma \cdot n$ p-squares that are distinct as strings (improving the $2(\sigma!)^2 n$ upper bound from \cite{DBLP:journals/tcs/KociumakaRRW16}). We show that all non-equivalent p-squares can be computed efficiently; our approach also extends to p-squares that are distinct as strings. See \cref{sec:2}.

\begin{theorem}\label{thm:main2}
All non-equivalent p-square substrings in a string of length $n$ over alphabet of size $\sigma$ can be computed in $\Oh(n \sigma \log \sigma)$ time.

All p-square substrings that are distinct as strings can be computed in $\Oh(n \sigma \log \sigma+\output)$ time, where $\output$ is the number of p-squares reported.
\end{theorem}

By the above discussion, all p-square substrings that are distinct as strings can be computed in $\Oh(n (\sigma+1)!)$ time.
The key ingredient in the algorithm behind \cref{thm:main2} is a relation between p-squares and $\sigma$-mismatch squares (and uniform $\sigma$-runs). 

\subsection{Generalised Squares under Substring Consistent Equivalence Relations}
We then show an alternative algorithm for counting p-squares whose complexity does not depend on the alphabet size $\sigma$. The algorithm is stated for squares under any \emph{substring consistent equivalence relation} (\emph{SCER} in short). An SCER is a relation $\approx$ on strings such that $X \approx Y$ implies that (1) $|X|=|Y|$ and (2) $X[i\dd j] \approx Y[i\dd j]$ for all $1 \le i \le j \le |X|$; see~\cite{DBLP:journals/tcs/MatsuokaAIBT16}. Known examples of SCERs include parameterized matching, order-preserving matching \cite{DBLP:journals/tcs/KimEFHIPPT14,DBLP:journals/ipl/KubicaKRRW13}, Cartesian tree matching \cite{DBLP:journals/tcs/ParkBALP20}, and palindrome pattern matching~\cite{DBLP:journals/tcs/IIT13}. Two strings $X$ and $Y$ \emph{order-preserving match} if $|X|=|Y|$, sets $\alph(X)$ and $\alph(Y)$ are totally ordered, and there is a bijection $f : \alph(X) \mapsto \alph(Y)$ that is increasing (i.e., if $x<y$ then $f(x)<f(y)$) such that $f(X)=Y$. Strings $X$, $Y$ \emph{Cartesian-tree match} if they have the same shape of a Cartesian tree (cf.\ \cite{10.1145/358841.358852}). Finally, strings $X$ and $Y$ \emph{palindrome match} if for all $1 \le i \le j \le |X|$, $X[i \dd j]$ is a palindrome if and only if $Y[i \dd j]$ is a palindrome.

An \emph{$\approx$-square} is a string $XY$ such that $X \approx Y$ under SCER $\approx$. Thus p-squares and order-preserving squares~\cite{DBLP:journals/tcs/KociumakaRRW16} (op-squares, in short) as well as Cartesian-tree squares~\cite{DBLP:journals/tcs/ParkBALP20} (CT-squares) and squares in the sense of palindrome matching (palindrome-squares) are $\approx$-squares for the respective SCERs $\approx$. See \cref{fig:3gensq} for an example.

\begin{figure}[htpb]
\begin{center}
\begin{tikzpicture}
    \foreach \c [count=\x] in {1,3,2,2,4,3,4,4,1,2,3,2,3,,,} {
        \node[above] (c\x) at (\x*0.35,0) {\c};
    }
    \node[above] at (-0.2,0) {$S=$};

    \foreach \a/\b/\c/\label in {1/4/7/$X$, 8/11/14/$Z$} {
      \draw ($(c\a.south)+(-0.175cm,0)$)--+(0,-0.1cm);
      \draw ($(c\b.south)+(-0.175cm,0)$)--+(0,-0.1cm);
      \draw ($(c\c.south)+(-0.175cm,0)$)--+(0,-0.1cm);
      \draw ($(c\a.south)+(-0.175cm,-0.1cm)$)--($(c\c.south)+(-0.175cm,-0.1cm)$);
      \node[below] at ($(c\b.south)+(-0.175cm,-0.1cm)$) {\label};
    }
    \foreach \a/\b/\c/\label in {3/7/11/$Y$} {
      \draw ($(c\a.north)+(-0.175cm,0)$)--+(0,0.1cm);
      \draw ($(c\b.north)+(-0.175cm,0)$)--+(0,0.1cm);
      \draw ($(c\c.north)+(-0.175cm,0)$)--+(0,0.1cm);
      \draw ($(c\a.north)+(-0.175cm,0.1cm)$)--($(c\c.north)+(-0.175cm,0.1cm)$);
      \node[above] at ($(c\b.north)+(-0.175cm,0.1cm)$) {\label};
    }
\end{tikzpicture}
\end{center}
\caption{For string $S=1322434412323$,
$X=132\ 243$ is an op-square (hence, automatically, a p-square and a CT-square),
$Y=2243\ 4412$ is a p-square,
and $Z=412\ 323$ is a CT-square. Among the three substrings, only $Z$ is \emph{not} a palindrome-square, as its second half is a palindrome whereas its first half is not.}
\label{fig:3gensq}
\end{figure}

In a prefix consistent equivalence relation (PCER), condition (2) of an SCER only needs to hold for $i=1$. An SCER is a PCER. We say that $\E:\Sigma^* \rightarrow \mathbb{Z}$ is an \emph{encoding function} if
\begin{equation}\label{eq:E}
X \approx Y \quad\Leftrightarrow\quad (\,|X|=|Y|\ \land\ \forall_{i \in [1 \dd |X|]}\,\E(X[1 \dd i])=\E(Y[1 \dd i])\,).
\end{equation}

\begin{observation}\label{obs:general}
For every PCER $\approx$ on strings over integer alphabet there exists an encoding function.
\end{observation}
\begin{proof}
Let $\E(X)$ be defined as the number of the equivalence class under $\approx$ of $X$ among all strings of length $|X|$. Let us verify that $\E$ satisfies equivalence~\eqref{eq:E}. $(\Rightarrow)$ If $X \approx Y$, then, by definition, $|X|=|Y|$ and $X[1 \dd i] \approx Y[1 \dd i]$ for all $i \in [1 \dd |X|]$. Then indeed $\E(X[1 \dd i])=\E(Y[1 \dd i])$. $(\Leftarrow)$ Taking $i=|X|=|Y|$, we obtain $\E(X)=\E(Y)$. As $|X|=|Y|$, we have $X \approx Y$.
\end{proof}

The encoding function defined in the proof of \cref{obs:general} could be hard to compute efficiently for a particular PCER $\approx$. It is also stronger, as it satisfies $X \approx Y\ \Leftrightarrow\ (\,|X|=|Y|\ \land\ \E(X)=\E(Y))$. Luckily, for the aforementioned known SCERs efficient encoding functions are known, as shown in the following \cref{ex:particular}.

For an encoding function $\E$, let $\pi^\E(n)$, $\rho^\E(n)$ be integer sequences such that for a string $T$ of length $n$ over integer alphabet, after $\pi^\E(n)$ preprocessing time, $\E(T[i \dd j])$ for any $i,j$ can be computed in $\rho^\E(n)$ time. 

\begin{example}\label{ex:particular}
For parameterized matching, there exists a classic $\mathit{prev}$-encoding, see \cite{DBLP:journals/jcss/Baker96}, that is an encoding function:
\[\mathit{prev}(X)=\left.
  \begin{cases}
    |X|-i & X[i]=X[|X|]\text{ and }X[j]\ne X[|X|]\text{ for all }j \in [i+1 \dd |X|) \\
    0 & \text{otherwise}
  \end{cases}
  \right.\]
The $\mathit{prev}$-encodings of all prefixes of a string $T$ can be computed by sorting all pairs $(T[i],i)$. If $T$ is over an integer alphabet, this is performed in $\pi^{\mathit{prev}}(n)=\Oh(n)$ time. Then $\mathit{prev}(T[i \dd j])$ can be computed from $\mathit{prev}(T[1 \dd j])$ in $\rho^{\mathit{prev}}(n)=\Oh(1)$ time.

For order-preserving matching, one can use an encoding as pairs $(\alpha(X),\beta(X))$. Here
\[\alpha(X)\text{ is the largest }j<|X|\text{ such that }X[j]=\max\{X[k]\ :\ k \in [1\dd |X|),\ X[k] \le X[|X|]\},\]
and if there is no such $j$, then $\alpha(X)=0$. Similarly,
\[\beta(X)\text{ is the largest }j<|X|\text{ such that }X[j]=\min\{X[k]\ :\ k \in [1 \dd |X|),\ X[k]\ge X[|X|]\},\]
and $\beta(X)=0$ if no such $j$ exists; see~\cite{DBLP:journals/tcs/CrochemoreIKKLP16,DBLP:journals/tcs/KimEFHIPPT14,DBLP:journals/ipl/KubicaKRRW13}.
As we require the encoding function to return integers, we can set $\E(X)=\alpha(X) \cdot |X|+\beta(X)$, since $\alpha(X),\beta(X) \in [0 \dd |X|)$. Then \cite[Lemma 4]{DBLP:journals/tcs/CrochemoreIKKLP16} implies that $\E$ is an encoding function for order-preserving matching and \cite[Lemma 24]{DBLP:journals/tcs/CrochemoreIKKLP16} gives $\pi^\E(n)=\Oh(n\sqrt{\log n})$, $\rho^\E(n)=\Oh(\log n/\log \log n)$.

For Cartesian-tree matching, \cite[Theorem 1]{DBLP:journals/tcs/ParkBALP20} shows that the following \emph{parent-distance} representation:
\[\mathit{PD}(X)=\left.
  \begin{cases}
    |X|-\max_{1 \le j < |X|} \{j\,:\,X[j] \le X[|X|]\} & \text{if such $j$ exists} \\
    0 & \text{otherwise}
  \end{cases}
  \right.\]
is an encoding function. Encodings of all prefixes of a string $T$ can be computed in $\pi^{\mathit{PD}}(n)=\Oh(n)$ time using a folklore nearest smaller value algorithm. In \cite[Section 5.2]{DBLP:journals/tcs/ParkBALP20} it is noted that $\mathit{PD}(T[i \dd j])$ can be computed from $\mathit{PD}(T[1 \dd j])$ in $\rho^{\mathit{PD}}(n)=\Oh(1)$ time. 

By $X^R$ we denote the reverse of string $X$. In \cite[Lemma 2]{DBLP:journals/tcs/IIT13} it is shown that the length of the longest suffix palindrome of $X$, formally:
\[\mathit{Lpal}(X)=\max\{|X|-k+1\,:\,X[k \dd |X|] = X[k \dd |X|]^R\},\]
is an encoding function for palindrome matching. 

Computing encodings of substrings $\mathit{Lpal}(T[i \dd j])$ can be reduced in linear time to 2D range successor queries as follows.

In the 2D range successor problem we are given $n$ points in an $n \times n$ grid and we are to answer queries that, given a rectangle in the grid, report a point in the rectangle with the smallest first coordinate, if any. Such queries can be answered in $\Oh(\log \log n)$ time after $\Oh(n \sqrt{\log n})$ preprocessing~\cite{DBLP:conf/esa/Gao0N20}.

In the reduction, we compute all maximal palindromes in $T$ in $\Oh(n)$ time using Manacher's algorithm~\cite{DBLP:journals/jacm/Manacher75}. Let us show how we deal with odd-length palindromes; the approach for even-length palindromes is analogous. For a maximal palindrome $T[c-r \dd c+r]$ with $c \in [1 \dd n]$ and $r \ge 0$, we create a point $(c,c+r)$. To compute $\mathit{Lpal}(T[i \dd j])$, it suffices to find the point in the rectangle $[\ceil{(i+j)/2}\dd j] \times [j \dd \infty)$ with the smallest first coordinate. If $x$ is the sought coordinate, the longest odd-length suffix palindrome of $T[i \dd j]$ has length $2(j-x)+1$.

By \cite{DBLP:conf/esa/Gao0N20}, we get $\pi^\mathit{Lpal}(n)=\Oh(n \sqrt{\log n})$ and $\rho^\mathit{Lpal}(n)=\Oh(\log \log n)$.
\end{example}

In \cref{sec:1} we show the next theorem on efficient counting of $\approx$-squares.

\begin{theorem}\label{thm:main1}
Let $\approx$ be an SCER and $\E$ be its encoding function. The number of non-equivalent $\approx$-square substrings in a length-$n$ string can be computed in $\Oh(\pi^\E(n)+n\rho^\E(n)+n\log n)$ time.

The number of $\approx$-square substrings in a length-$n$ string that are distinct as strings can be computed in the same time complexity.
\end{theorem}

Consequently, in a length-$n$ string, the numbers of non-equivalent p-squares, op-squares, CT-squares, and palindrome-squares, as well as the numbers of the respective squares that are distinct as strings can be computed in $\Oh(n \log n)$ time, as the respective SCERs enjoy encoding functions with $\pi^\E(n)=\Oh(n \log n)$ and $\rho^\E(n)=\Oh(\log n)$ as shown in \cref{ex:particular}. In particular, the algorithm of \cref{thm:main1} in the case that $\sigma=\omega(\log n)$ improves upon the $\Oh(n\sigma)$-time algorithm of Gawrychowski, Ghazawi, and Landau~\cite{DBLP:conf/cpm/GawrychowskiGL23} for reporting (thus, counting) op-squares that are distinct as strings. Moreover, the obtained complexity for counting non-equivalent p-squares (p-squares that are distinct as strings, respectively) is better than the one in \cref{thm:main2} if $\sigma=\omega(\log n/\log \log n)$ ($\sigma=\Omega(\log \log n)$, respectively).

\cref{thm:main1} involves constructing a $\approx$-suffix tree data structure (see \cref{sec:ST} for a definition).

\section{Upper Bound on the Number of \texorpdfstring{$k$}{k}-Runs}\label{sec:KK++}
This section introduces an upper bound on the number of $k$-runs (proof of \cref{thm:KK++}) which implies their efficient computation by the algorithm of Kolpakov and Kucherov~\cite{DBLP:conf/esa/KolpakovK01,DBLP:journals/tcs/KolpakovK03}.
A position $i$ in string $T$ is called an \emph{$\ell$-mismatching position} if $i \in [1\dd n-\ell]$ and $T[i] \ne T[i+\ell]$.
In particular, $T[i \dd i+2\ell)$ is a $k$-mismatch square if and only if $T[i \dd i+\ell)$ contains at most $k$ $\ell$-mismatching positions.
Let us formally define a uniform $k$-run.

\begin{definition}

A substring $T[a \dd b)$ is called a \emph{uniform $k$-run of period $\ell$} if $b-a \ge 2\ell$ and the sets $\{j \in [i \dd i+\ell)\,:\,T[j] \ne T[j+\ell]\}$ of $\ell$-mismatching positions for all $i \in [a \dd b-2\ell]$ have cardinality at most $k$ and are all the same.
We are interested in maximal uniform $k$-runs.
\end{definition}

A \emph{gapped repeat} in a string $T$ is a fragment of the form $UVU$, for $|U|>0$. Fragments denoted by $U$ are called \emph{arms} of the gapped repeat (left and right arm). The \emph{period} of the gapped repeat is defined as $|UV|$. An $\alpha$-gapped repeat (for $\alpha \ge 1$) in $T$ is a gapped repeat $UVU$ such that $|UV| \le \alpha|U|$. An ($\alpha$-)gapped repeat is called \emph{maximal} if its arms cannot be extended simultaneously with the same character to the right or to the left. A maximal ($\alpha$-)gapped repeat will be called an ($\alpha$-)MGR, for short. Gawrychowski, I, Inenaga, Köppl, and Manea~\cite{DBLP:journals/mst/GawrychowskiIIK18} showed that, for a real $\alpha \ge 1$, a string of length $n$ contains $\Oh(n\alpha)$ $\alpha$-MGRs.

A \emph{generalised run} is a pair $(T[x \dd y),p)$ such that $p$ satisfies $2p \le y-x$ and (1) $x=1$ or $T[x-1 \dd y)$ does not have period $p$, and (2) $y=|T|+1$ or $T[x \dd y]$ does not have period $p$. That is, a generalised run is a 0-run. A run is a generalised run in which $p$ is the smallest period of $T[x \dd y)$. A string of length $n$ contains less than $n$ runs and $\Oh(n)$ generalised runs~\cite{DBLP:journals/siamcomp/BannaiIINTT17}.

\begin{definition}
We say that an MGR $UVU=T[x \dd y)$ with period $\ell=|UV|$ in $T$ \emph{induces} a uniform $k$-run $T[a \dd b)$ of period $\ell$ if $[x \dd y-\ell) \cap [a \dd b-\ell) \ne \emptyset$.

We further say that a generalised run $(T[x \dd y),\ell)$ \emph{induces} a uniform $k$-run $T[a \dd b)$ of period $\ell$ if $[x \dd y-\ell) \cap [a \dd b-\ell) \ne \emptyset$.
\end{definition}

Let us emphasize that the uniform $k$-run does not need to be a substring of the MGR or generalised run.
Intuitively, an MGR induces a uniform $k$-run if the left half of at least one $k$-mismatch square of the $k$-run contains a position of the left arm of the MGR; see \cref{fig:kruns_gr}.
If an MGR or a generalised run induces a uniform $k$-run, we say that the $k$-run is \emph{induced} by the MGR or run, respectively. We also say that the MGR or generalised run induces all $k$-mismatch squares in the uniform $k$-run that it induces.

\begin{figure}[htpb]
\centering
\include{_fig_kruns_gr}
\caption{String $T[1 \dd 26]$ from \cref{fig:kruns} with 8-mismatching positions shown in red. The string contains three uniform 2-runs with period 8, $T[1\dd 18]$, $T[5 \dd 22]$, and $T[8 \dd 24]$ (in \cref{fig:kruns}, the last two form a single 2-run). They are all induced by a 3-MGR $UVU$ for $U=T[8 \dd 10]=T[16 \dd 18]$ and $V=T[11\dd 15]$. In the proof of \cref{lem:MGR_packages} for this MGR, we have $L=\{7,4,0\}$ and $R=\{11,15,17\}$.}
\label{fig:kruns_gr}
\end{figure}

\begin{lemma}\label{lem:MGR_packages}
An MGR induces at most $2k+1$ uniform $k$-runs.
\end{lemma}
\begin{proof}
Consider an MGR $T[x \dd y)$ with period $\ell$. Let $R$ be the set of $k+1$ smallest $\ell$-mismatching positions in $[y-\ell \dd y)$ in $T$; if there are no such $k+1$ positions, we select all $\ell$-mismatching positions in $[y-\ell \dd y)$ in $T$ to the set $R$. Similarly, let $L$ be the set of $k+1$ largest $\ell$-mismatching positions in $[x-\ell+1 \dd x)$ in $T$; if there are no $k+1$ such positions, we select all $\ell$-mismatching positions in $[x-\ell+1 \dd x)$. If $|R| \le k$, we insert $y$ to $R$. If $|L| \le k$, we insert $x-\ell$ to $L$. This way, if a $k$-mismatch square $T[i \dd i+2\ell)$ is induced by the MGR $T[x \dd y)$, then $i \in I$ for $I=[\min L + 1 \dd \max R - \ell]$ (as otherwise the left half of the square would contain more than $k$ $\ell$-mismatching positions or would not touch $T[x \dd y-\ell)$).

Let us consider a window of length $\ell$ sliding left-to-right over $T$ with the left end point of the window located in interval $I$. There are at most $2k$ moments when the number of $\ell$-mismatching positions in the window changes. Between them, we obtain uniform $k$-runs if the number of $\ell$-mismatching positions in the window is at most $k$.
\end{proof}

An analogous lemma for generalised runs is obtained with the same proof (replacing ``MGR'' by ``generalised run''); see \cref{fig:kruns_run}.

\begin{figure}[htpb]
\centering
\include{_fig_kruns_run}
\caption{String $T[1 \dd 26]$ with its 6-mismatching positions shown in red. A generalised run $T[6 \dd 21]$ with period 6 (rectangle) induces five uniform 2-runs with period 6 shown above. In this example, the generalised run corresponds to a run with period 3.}
\label{fig:kruns_run}
\end{figure}

\begin{lemma}\label{lem:run_packages}
A generalised run induces at most $2k+1$ uniform $k$-runs.
\end{lemma}

A uniform $k$-run can be induced by several MGRs and generalised runs.

\begin{example}
The middle uniform 2-run $T[5 \dd 22]$ in \cref{fig:kruns_gr} is induced also by two other MGRs with period 8, $T[5 \dd 14]$ (with arms $T[5 \dd 6]=T[13 \dd 14]=\mathtt{aa}$) and $T[12 \dd 22]$ (with arms $T[12 \dd 14]=T[20 \dd 22]=\mathtt{baa}$).
\end{example}

For a gapped repeat $UVU$, the fraction $\frac{|UV|}{|U|}$ is called the \emph{gap ratio}. That is, the gap ratio is the minimum $\alpha$ such that $UVU$ is an $\alpha$-gapped repeat. We assign to each gapped repeat $UVU$ in $T$ a \emph{weight} equal to the reciprocal of its gap ratio, that is, $\frac{|U|}{|UV|}$.

\begin{lemma}\label{lem:packages_MGR}
If $\ell \ge 4k$ and $T[a \dd b)$ is a uniform $k$-run of period $\ell$, then it is induced by a generalised run or by some number of $(2k+2)$-MGRs of total weight at least $\frac14$.
\end{lemma}
\begin{proof}
Let $p_1,\ldots,p_t$ be all $\ell$-mismatching positions among $[a \dd b-\ell)$, listed in an increasing order. We define sentinel $p_0$ being the maximum $\ell$-mismatching position that is smaller than $a$; if there is no such position, we set $p_0=0$. Similarly, we define sentinel $p_{t+1}$ as the minimum $\ell$-mismatching position that is at least $b-\ell$; if no such position exists, we set $p_{t+1}=n+1-\ell$.

Let us fix $s \in [0 \dd t]$. We have $U_s:=T[p_s+1 \dd p_{s+1}) = T[p_s+\ell+1 \dd p_{s+1}+\ell)$, so $W_s:=T[p_s+1 \dd p_{s+1}+\ell)$ has period $\ell$. Moreover, we have (1) $p_s=0$ or $T[p_s] \ne T[p_s+\ell]$ and (2) $p_{s+1}+\ell=n+1$ or $T[p_{s+1}] \ne T[p_{s+1}+\ell]$. If $0<|U_s|<\ell$, then $W_s$ is an MGR with arms equal to $U_s$. If $|U_s| \ge \ell$, $W_s$ is a generalised run. The MGR or generalised run induces $T[a \dd b)$ unless $p_{s+1}=a$ or $p_s=b-\ell-1$.

Assume that the uniform $k$-run $T[a \dd b)$ is not induced by a generalised run. If a string $W_s$, for $s \in [0 \dd t]$, is not a $(2k+2)$-MGR of period $\ell$, the length of its left arm $U_s$ is smaller than $\frac{\ell}{2k+2}$. Hence, for strings $W_s$ that are not $(2k+2)$-MGRs, the total length of their left arms $U_s$ is at most $\frac{\ell}{2}$. The remaining $(2k+2)$-MGRs among $W_s$ have total length of left arms at least $(b-a-\ell)-\frac{\ell}{2}-k \ge \frac{\ell}{2}-k$ (accounting for the at most $k$ $\ell$-mismatching positions $p_1,\ldots,p_t$) which is at least $\frac{\ell}{4}$ by the assumption of the lemma. This means that the total weight of $(2k+2)$-MGRs that induce the $k$-run, defined as the sum of the lengths of their left arms divided by their periods, all equal to $\ell$, is at least $\frac14$.
\end{proof}

We show a theorem that implies \cref{thm:KK++}.

\begin{theorem}\label{thm:KK++uni}
A string of length $n$ contains $\Oh(nk\log k)$ uniform $k$-runs.
\end{theorem}
\begin{proof}
Let $T$ be a string of length $n$.
Let us consider a bipartite graph $G=(V_1 \cup V_2,E)$ such that vertices in $V_1$ are uniform $k$-runs in $T$, vertices in $V_2$ are $(2k+2)$-MGRs and generalised runs in $T$, and there is an edge $uw \in E$, for $u \in V_1$ and $w \in V_2$, if $k$-run $u$ is induced by the MGR or generalised run $w$. Our goal is to bound $|V_1|$ from above.

Let us remove from $G$ all vertices in $V_1$ that are uniform $k$-runs with period at most $4k$. There are at most $4kn$ such $k$-runs (as there are at most $4kn$ substrings of $T$ of length at most $4k$).
Let us also remove from $G$ all vertices in $V_2$ that are generalised runs and all vertices in $V_1$ that are adjacent to at least one of the removed vertices in $V_2$. By \cite{DBLP:journals/siamcomp/BannaiIINTT17}, this way at most $1.5n$ vertices are removed from $V_2$, so by \cref{lem:run_packages}, $\Oh(nk)$ vertices are removed from $V_1$. Let $G'=(V'_1\cup V'_2,E')$ be the remaining graph. We assume that each edge has a weight equal to the weight of the MGR in $V'_2$ it is incident to. In order to bound $|V'_1|$, we will consider the \emph{total weight of edges} in $E'$.

For each $\alpha \in [2 \dd 2k+2]$, let $x_\alpha$ be the number of $\alpha$-MGRs in $T$ that are not $(\alpha-1)$-MGRs. Each such MGR has weight at most $\frac{1}{\alpha-1}$ and $2k+1 \le 3k$ edges incident to it (cf.\ \cref{lem:MGR_packages}). Thus the sum of weights of edges in $E'$ is bounded from above by
\begin{equation}\label{eq:ub}
\sum_{\alpha=2}^{2k+2} \frac{3kx_\alpha}{\alpha-1}
\end{equation}
The improved upper bound on the number of $\alpha$-MGRs of I and Köppl \cite{DBLP:journals/tcs/IK19} implies that
\begin{equation}\label{eq:cond}
\sum_{i=2}^\alpha x_i < 13n\alpha \quad\text{for every }\alpha \in [2 \dd 2k+2]
\end{equation}
We show how to bound the number \eqref{eq:ub} in general.
\begin{claim}
For every sequence $(x_2,\ldots,x_{2k+2})$ of non-negative integers that satisfies the condition \eqref{eq:cond}, the value \eqref{eq:ub} is $\Oh(nk\log k)$.
\end{claim}
\begin{proof}
As long as there exists $\alpha \in [3 \dd 2k+2]$ such that $x_\alpha>13n$, we select $i \in [2 \dd \alpha)$ such that $x_i < 13n$ or $i=2$ and $x_i < 26n$, decrement $x_\alpha$ and increment $x_i$. Such $i$ exists by \eqref{eq:cond} and the pigeonhole principle. The operation does not change any lhs in \eqref{eq:cond} and increases \eqref{eq:ub}.

In the end, we have $x_2 \le 26n$ and $x_\alpha \le 13n$ for all $\alpha \in [3 \dd 2k+2]$. Hence, \eqref{eq:ub} is bounded as:
\[\sum_{\alpha=2}^{2k+2} \frac{3kx_\alpha}{\alpha-1}\ \le \ 78k+39nk\sum_{\alpha=2}^{2k+1} \frac{1}{\alpha}\ \le\ 78k+39nk(\ln(2k+1)+1)\ =\ \Oh(nk \log k).\qedhere\]
\end{proof}

By the claim, the sum of weights of edges in $G'$ is $\Oh(nk\log k)$. By \cref{lem:packages_MGR}, for each uniform $k$-run with period $\ell \ge 4k$ that is not induced by a generalised run, the total weight of edges incident to it is at least $\frac14$. This concludes that $|V'_1|=\Oh(nk \log k)$, so the total number of uniform $k$-runs (across all periods $\ell \in [1 \dd \floor{n/2}]$) is $|V_1|=\Oh(nk \log k)$.
\end{proof}

A $k$-run of period $\ell$ contains a prefix being a uniform $k$-run of period $\ell$, constructed by taking $k$-mismatch squares of length $2\ell$ at subsequent positions as long as their sets of $\ell$-mismatching positions are the same. By maximality, no two $k$-runs with equal period start at the same position, so this way each $k$-run produces a different uniform $k$-run. Thus the number of $k$-runs in $T$ does not exceed the number of uniform $k$-runs, which proves \cref{thm:KK++}. By the same argument, each $k$-repetition implies a unique uniform $k$-run and so the number of $k$-repetitions is $\Oh(nk \log k)$. (We refer the reader to the definition of $k$-repetition in \cite{DBLP:conf/esa/KolpakovK01,DBLP:journals/tcs/KolpakovK03,Kucherov2014} as this notion is not used below.)

\section{\texorpdfstring{$\approx$}{≈}-Suffix Trees and Their Applications}\label{sec:ST}
We introduce useful tools for the next two sections.
Cole and Hariharan~\cite{DBLP:journals/siamcomp/ColeH03a} defined a \emph{quasi-suffix collection} as a collection of strings $S_1,S_2,\ldots,S_n$ that satisfies the following conditions:
\begin{enumerate}
\item\label{it1qs} $|S_1| = n$ and $|S_i| = |S_{i-1}| - 1$.
\item\label{it2qs} No $S_i$ is a prefix of another $S_j$.
\item\label{it3qs} Suppose strings $S_i$ and $S_j$ have a common prefix of length $\ell > 0$. Then $S_{i+1}$ and $S_{j+1}$ have a common prefix of length at least $\ell - 1$.
\end{enumerate}
A quasi-suffix collection is specified implicitly by a character oracle that given $i$, $j$ returns $S_i[j]$. They obtain the following result.

\begin{theorem}[\cite{DBLP:journals/siamcomp/ColeH03a}]\label{thm:CH}
The compacted trie of a quasi-suffix collection of $n$ strings can be constructed in $\Oh(n)$ expected time assuming that the character oracle works in $\Oh(1)$ time.
\end{theorem}

\newcommand{\Code}{\mathit{Code}}
For a string $T$ of length $n$, the \emph{$\approx$-suffix tree} of $T$ is a compacted trie of strings
$\Code^\E(T[i \dd n])\#$, $i \in [1 \dd n+1]$, where 
$\Code^\E(X)=\E(X[1])\E(X[1 \dd 2])\cdots \E(X)$ for a string $X$ and $\#$ is an end-marker that is not an encoding of any string.
We extend the sequences $\pi$, $\rho$ so that for any string $T \in [0 \dd \sigma)^n$, after $\pi^\E(n,\sigma)$ preprocessing time, $\E(T[i \dd j])$ for any $i,j$ can be computed in $\rho^\E(n,\sigma)$ time.
By $\gamma^\E(n,\sigma)$ we denote an upper bound on the total number of distinct characters in the strings stored in the $\approx$-suffix tree of a string in $[0 \dd \sigma)^n$.
\cref{thm:CH} implies the following corollary. (A similar result, but only for order-preserving equivalence, was shown in \cite[Lemma 16]{DBLP:conf/cpm/GawrychowskiGL23}.)

\begin{corollary}\label{cor:CH}
Let $\approx$ be an SCER and $\E$ be its encoding function. The $\approx$-suffix tree of a string in $[0 \dd \sigma)^n$ can be constructed in worst case time:
\[\Oh(\pi^\E(n,\sigma)+n\rho^\E(n,\sigma)+\min\{n \log^2 \log n/\log\log\log n,\,n\gamma^\E(n,\sigma)\}).\]
\end{corollary}
\begin{proof}
Let us verify that strings $\Code^\E(T[i \dd n])\#$ satisfy the conditions \ref{it1qs}-\ref{it3qs} of a quasi-suffix collection. Condition \ref{it1qs} for $n+1$ is obvious. Condition \ref{it2qs} follows due to the end-marker. As for condition \ref{it3qs}, assume $\LCP(S_i,S_j)=\ell>0$ and let $i \ne j$ as otherwise the conclusion is trivial. Then $\ell < |S_i|,|S_j|$ because of the end-marker. By equivalence \eqref{eq:E}, we have $T[i \dd i+\ell) \approx T[j \dd j+\ell)$. Because $\approx$ is an SCER, we have $T[i+1 \dd i+\ell) \approx T[j+1 \dd j+\ell)$. Hence, $\LCP(S_{i+1},S_{j+1})\ge\ell-1$ by equivalence \eqref{eq:E}.

An oracle for the quasi-suffix collection $S_i$ answers queries in $\Oh(\rho^\E(n,\sigma))$ time after $\Oh(\pi^\E(n,\sigma))$ preprocessing. The only source of randomness in the algorithm behind \cref{thm:CH} is the need to maintain, for each explicit node of the current tree, a dictionary indexed by the next character on an outgoing edge. If we store the respective characters per each node in a dynamic predecessor data structure of Andersson and Thorup~\cite{DBLP:journals/jacm/AnderssonT07}, the total space remains $\Oh(n)$ and each predecessor query is answered in $\Oh(\log^2 \log n/\log\log\log n)$ time in the worst case. Alternatively, one can store all the children of a node in a list, to achieve $\Oh(\gamma^\E(n,\sigma))$ space per an explicit node and $\Oh(\gamma^\E(n,\sigma))$ time for a predecessor query. The complexity follows.
\end{proof}

\begin{remark}
For SCERs mentioned in \cref{ex:particular}, we obtain the following time complexities for constructing $\approx$-suffix trees in the respective settings: $\Oh(n \log^2 \log n/\log\log\log n)$ for parameterized matching and Cartesian tree matching, $\Oh(n \sqrt{\log n})$ for palindrome matching (that matches the complexity from \cite{DBLP:journals/tcs/IIT13}), and $\Oh(n \log n /\log \log n)$ for order-preserving matching (a faster, $\Oh(n \sqrt{\log n})$-time construction was proposed in \cite{DBLP:journals/tcs/CrochemoreIKKLP16}).
\end{remark}

For two strings $X$ and $Y$, by $\LCP^\approx(X,Y)$ we denote $\max\{\ell \ge 0\,:\, X[1 \dd \ell] \approx Y[1 \dd \ell]\}$. As in the case of standard suffix trees, by equivalence \eqref{eq:E}, having an $\approx$-suffix tree of $T$ and the data structure answering lowest common ancestor queries for nodes in $\Oh(1)$ time after $\Oh(n)$ preprocessing~\cite{DBLP:journals/siamcomp/HarelT84,DBLP:conf/latin/BenderF00}, we can answer $\LCP^\approx$ queries about pairs of substrings of $T$ in $\Oh(1)$ time.

The longest previous $\approx$-factor array $\LPF^\approx$ for a string $T$ is an array such that
\[\LPF^\approx[i]=\max\{\ell \ge 0\,:\,T[i \dd i+\ell)\approx T[j \dd j+\ell)\text{ for some }j \in [1 \dd i)\}.\]
Computing this array can be stated in terms of the $\approx$-suffix tree: for a leaf with index $i$ of the tree we are to find a leaf with index $j < i$ such that the lowest common ancestor of the two leaves is as low as possible. The actual value $\LPF^\approx[i]$ is then the (weighted) depth of this ancestor. In \cite[Theorem 16]{DBLP:journals/tcs/CrochemoreIKKLP16} it was shown that this problem, stated for an arbitrary rooted (weighted) tree with $n$ leaves, can be solved in $\Oh(n)$ time. Thus, the $\LPF^\approx$ array for a length-$n$ string can be computed in $\Oh(n)$ time if the $\approx$-suffix tree is available.

\section{Counting p-Squares in \texorpdfstring{$\Oh(n \sigma \log \sigma)$}{O(nσ log σ)} Time}\label{sec:2}
In this section we consider $T \in [0 \dd \sigma)^n$ and $\approx$ denotes the relation of parameterized matching.
We use the following function $\EE$ for $\approx$ (see \cref{ex:EE}):
\[\EE(X)=|\alph(U)|\text{ where }U\text{ is the longest suffix of }X[1 \dd |X|)\text{ without letter }X[|X|].\]

\begin{example}\label{ex:EE}
$\EE(X)$ for $X=[1, 2, 1, 1, 2, 3, 2, 1, 4]$ is as follows:
\begin{center}\begin{tabular}{r|lllllllll}
$i$ & \small 1 & \small 2 & \small 3 & \small 4 & \small 5 & \small 6 & \small 7 & \small 8 & \small 9 \\ \hline
$X[i]$ & 1 & 2 & 1 & 1 & 2 & 3 & 2 & 1 & 4 \\
$\EE(X[1\dd i])$ & 0 & 1 & 1 & 0 & 1 & 2 & 1 & 2 & 3 \\
\end{tabular}\end{center}
\end{example}

A similar function was used in the context of parameterized matching in \cite{DBLP:journals/tcs/KociumakaRRW16}. \cref{lem:EE_enc} below shows that $\EE$  is an encoding function for parameterized matching and \cref{lem:EE_enc2} shows that it can be computed efficiently.

The ($\Leftarrow$) part of the proof of \cref{lem:EE_enc} is similar to the proof of \cite[Lemma 5.4]{DBLP:journals/tcs/KociumakaRRW16}.

\begin{lemma}\label{lem:EE_enc}
$\EE$ is an encoding function for parameterized matching.
\end{lemma}
\begin{proof}
We show that $\EE$ satisfies the equivalence \eqref{eq:E} (as $\E$). The proof of part $(\Leftarrow)$ is by induction on $|X|$. For length 0 the claim is trivial. Suppose that it holds for all lengths strictly smaller than $x$ and let strings $X$, $Y$ with $|X|=|Y|=x$ satisfy the right hand side of \eqref{eq:E}. For $X'=X[1 \dd x)$ and $Y'=Y[1 \dd x)$, $X' \approx Y'$ with some witness bijection $f:\alph(X') \mapsto \alph(Y')$.

We consider two cases. First, suppose that
\[\EE(X)=\EE(Y)\ge |\alph(X')|=|\alph(Y')|.\]
By definition of $\EE$, this implies that $X[x] \not\in \alph(X')$ and $Y[x] \not\in \alph(Y')$. Consequently, $f$ can be extended with $X[x] \mapsto Y[x]$ which yields the witness bijection for equivalence of $X$ and $Y$.

Next, suppose that
\[\EE(X)=\EE(Y)=k<|\alph(X')|=|\alph(Y')|.\]
This means that $X[x]$ is the $k$th element of $\alph(X')$ ordered according to the last occurrence in $X'$ and $Y[x]$ is the $k$th element of $\alph(Y')$ ordered according to the last occurrence in $Y'$. Since $X' \approx Y'$, the relative positions of these last occurrences are the same. Hence, we can set $f(X[x])=Y[x]$ and $f$ is the witness bijection for equivalence of $X$ and $Y$.

\smallskip
We proceed to the proof of part ($\Rightarrow$) of equivalence \eqref{eq:E}. If $X \approx Y$, then clearly $|X|=|Y|$. Let $i \in [1 \dd |X|]$. As $\approx$ is an SCER, we have $X[1 \dd i] \approx Y[1 \dd i]$. Again, we consider two cases. If $X[i] \not\in \alph(X[1 \dd i))$, then $Y[i] \not\in \alph(Y[1 \dd i))$ and
\[\EE(X[1 \dd i])=\EE(Y[1 \dd i])=|\alph(X[1 \dd i))|=|\alph(Y[1 \dd i))|\]
since $X[1 \dd i) \approx Y[1 \dd i)$.
If $X[i] \in \alph(X[1 \dd i))$, let $X[j]=X[i]$ be the last occurrence of $X[i]$ in $X[1 \dd i)$. Then $Y[j]=Y[i]$ and
\[\EE(X[1 \dd i])=\EE(Y[1 \dd i])=|\alph(X[j+1 \dd i))|=|\alph(Y[j+1 \dd i))|\]
as $X[j+1 \dd i) \approx Y[j+1 \dd i)$.
\end{proof}

\begin{lemma}\label{lem:EE_enc2}
We have $\pi^\EE(n,\sigma)=\Oh(n\sigma)$, $\rho^\EE(n,\sigma)=\Oh(\sigma)$, and $\gamma^\EE(n,\sigma)=\sigma+1$.
\end{lemma}
\begin{proof}
Let $T \in [0 \dd \sigma)^n$. By definition, $\EE(T[i \dd j]) \in [0 \dd \sigma)$ for any indices $i,j$. Hence, $\gamma^\EE(n,\sigma)=\sigma+1$ (including the sentinel). In order to compute encodings of substrings of $T$, we can store for every $c \in [0 \dd \sigma)$, the numbers of characters $c$ in respective prefixes of $T$. Moreover, for every index $i$, we store the position $\mathit{prev}[i] \in [0 \dd i)$ of the last occurrence of character $T[i]$ in $T[1 \dd i)$; $\mathit{prev}[i]=0$ if there is no such occurrence. 
The array $\mathit{prev}$ can be computed from left to right in $\Oh(n+\sigma)$ time by storing an array of size $\sigma$ of rightmost occurrences of each character. Then indeed $\pi^\EE(n,\sigma)=\Oh(n\sigma)$. When computing $\EE(T[i \dd j])$, we can compare the counts of every character $c$ at prefixes $T[1 \dd j)$ and $T[1 \dd i')$ 
where $i'=\max(\mathit{prev}[j]+1,i)$. Hence, $\rho^\EE(n,\sigma)=\Oh(\sigma)$.
\end{proof}

We define strings $\TR$, $\TL$ of length $n$ such that
\[\TR[i]=\EE(T[1 \dd i]),\quad \TL[i]=\EE((T[i \dd n])^R).\]
Our goal is to reduce reporting non-equivalent p-squares to computing uniform $k$-runs. We use two auxiliary lemmas.

\begin{lemma}\label{lem:comb2}
If $T[i \dd i+2\ell)$ is a p-square, then $\TR[i \dd i+2\ell)$ and $\TL[i \dd i+2\ell)$ are $\sigma$-mismatch squares.
\end{lemma}
\begin{proof}
We show a proof that $\TR[i \dd i+2\ell)$ is a $\sigma$-mismatch square; the proof that $\TL[i \dd i+2\ell)$ is a $\sigma$-mismatch square is symmetric.

Let $i_1,\ldots,i_t \in [0 \dd \ell)$, for $t \in [1 \dd \sigma]$, be the positions of leftmost occurrences of characters from $[0 \dd \sigma)$ in $T[i \dd i+\ell)$; if some character is not present in $T[i \dd i+\ell)$, it is not included in the sequence. Let $j \in [0 \dd \ell) \setminus \{i_1,\ldots,i_t\}$. Let $j' \in [0 \dd j)$ be the maximum index such that $T[i+j']=T[i+j]$, so $T[i+\ell+j']=T[i+\ell+j]$ and $T[i+\ell+j''] \ne T[i+\ell+j]$ for all $j'' \in [j'+1 \dd j)$ because $T[i \dd i+\ell) \approx T[i+\ell \dd i+2\ell)$. Then
\[\TR[i+j]\,=\,|\alph(T[i+j'+1 \dd i+j))|\,=\,|\alph(T[i+\ell+j'+1 \dd i+\ell+j))|\,=\,\TR[i+\ell+j]\]
by the fact that $\approx$ is an SCER.
Hence, $\TR[i \dd i+\ell)$ and $\TR[i+\ell \dd i+2\ell)$ have at most $t \le \sigma$ mismatches, as required.
\end{proof}

\begin{lemma}\label{lem:comb1}
If $T[i \dd i+2\ell)$ is a p-square and $\TR[i+\ell]=\TR[i+2\ell]$, then $T[i+1 \dd i+2\ell]$ is a p-square.
Similarly, if $T[i \dd i+2\ell)$ is a p-square and $\TL[i-1]=\TL[i+\ell-1]$, then $T[i-1 \dd i+2\ell-1)$ is a p-square.
\end{lemma}
\begin{proof}
Again, we only prove the first part of the lemma.
We have $T[i \dd i+\ell) \approx T[i+\ell \dd i+2\ell)$, so $T[i+1\dd i+\ell) \approx T[i+1+\ell \dd i+2\ell)$ since $\approx$ is an SCER. If
\[\TR[i+2\ell]=\TR[i+\ell] < |\alph(T[i+1 \dd i+\ell))|=|\alph(T[i+1+\ell\dd i+2\ell))|,\]
then $\TR[i+\ell]=\EE(T[i+1 \dd i+\ell])$, $\TR[i+2\ell]=\EE(T[i+1+\ell \dd i+2\ell])$, and $T[i+1 \dd i+\ell] \approx T[i+1+\ell \dd i+2\ell]$ by the fact that $\EE$ is an encoding function for parameterized matching (\cref{lem:EE_enc}). Otherwise,
\[T[i+\ell] \not\in\alph(T[i+1 \dd i+\ell))\quad\text{and}\quad T[i+2\ell] \not\in\alph(T[i+1+\ell \dd i+2\ell)),\]
so we immediately obtain $T[i+1 \dd i+\ell] \approx T[i+1+\ell \dd i+2\ell]$. In both cases, $T[i+1 \dd i+2\ell]$ is a p-square.
\end{proof}

If $T[a \dd b)$ is a uniform $k$-run with period $\ell$, then there is no $\ell$-mismatching position in $T[a+\ell \dd b-\ell)$, i.e., $T[a+\ell \dd b-\ell)=T[a+2\ell \dd b)$.
Hence, if $\TR[a \dd b)$ is a uniform $k$-run with period $\ell$ and $T[a \dd a+2\ell)$ is a p-square, then by \cref{lem:comb1}, each string $T[i \dd i+2\ell)$ for $i \in [a \dd b-2\ell]$ is a p-square.
With this intuition, we are ready to obtain the reduction.

\begin{lemma}\label{lem:red}
Let $T \in [0 \dd \sigma)^n$.
Reporting non-equivalent p-squares in $T$ reduces in $\Oh(n\sigma+r)$ time to computing uniform $\sigma$-runs in $\TR$ and $\TL$, where $r$ is the number of these $\sigma$-runs.

Reporting p-squares in $T$ that are distinct as strings reduces to the same problem in $\Oh(n\sigma+r+\mathsf{output})$ time, where $\mathsf{output}$ is the number of p-squares reported.
\end{lemma}
\begin{proof}
For a uniform $k$-run $T[a \dd b)$ of period $\ell$, we define its \emph{interval} as $[a \dd b-2\ell]$.
For each $\ell \in [1 \dd \floor{n/2}]$, let $\PR_\ell$ and $\PL_\ell$ be sets of intervals of uniform $\sigma$-runs of period $\ell$ in $\TR$ and $(\TL)^R$, respectively. Let $\PL'_\ell=\{[n-b-2\ell\dd n-a-2\ell]\,:\,[a\dd b] \in \PL_\ell\}$. Finally, let $\P_\ell=\left(\bigcup \PR_\ell\right) \cap \left(\bigcup \PL'_\ell\right)$.

\begin{claim}\label{clm:1}
If $T[i \dd i+2\ell)$ is a p-square in $T$, then $i \in \P_\ell$.
\end{claim} 
\begin{proof}
By \cref{lem:comb2}, if $T[i \dd i+2\ell)$ is a p-square in $T$, then $\TR[i \dd i+2\ell)$ and $\TL[i \dd i+2\ell)=(\TL)^R[n-i-2\ell \dd n-i)$ are $\sigma$-mismatch squares. Therefore, there exist intervals $[a \dd b] \in \PR_\ell$ and $[a' \dd b'] \in \PL_\ell$ such that $i \in [a \dd b]$ and $n-i-2\ell \in [a' \dd b']$. In particular, $i \in \bigcup \PR_\ell$. We have $[n-b'-2\ell \dd n-a'-2\ell] \in \PL'_\ell$ and $i \in [n-b'-2\ell \dd n-a'-2\ell]$, so $i \in \bigcup \PL'_\ell$.
\end{proof}

We store $\P_\ell$ as a union of non-empty intervals of the form $\{I \cap J\,:\,I \in \PR_\ell,\,J \in \PL'_\ell\}$. Let this representation be denoted as $\R_\ell$.

All the representations $\R_\ell$ can be computed from $\PR_\ell$ and $\PL'_\ell$ in $\Oh(n+\sum_\ell(|\PR_\ell|+|\PL_\ell|))$ total time. Let us bucket sort all endpoints of intervals in $\PR_\ell$ and $\PL'_\ell$, for each $\ell$. When processing the endpoints for a given $\ell$, we keep track of the number of intervals containing a given position. This counter never exceeds 2 as intervals in each of $\PR_\ell$ and $\PL'_\ell$ are pairwise disjoint. Whenever the counter reaches 2, for some endpoint $a$, it will drop at the next endpoint encountered, say $b$. Then $[a \dd b)$ is inserted into $\R_\ell$.

The next claim shows that, for each interval in $\R_\ell$, all positions correspond to p-squares of length $2\ell$ or none of them does.

\begin{claim}\label{clm:2}
Let $[a \dd b] \in \R_\ell$. For each $i \in [a \dd b]$, $T[i \dd i+2\ell)$ is a p-square if and only if $T[a \dd a+2\ell)$ is a p-square.
\end{claim}
\begin{proof}
Let us fix $i \in [a \dd b]$. Let $[a' \dd b'] \in \PR_\ell$ and $[a'' \dd b''] \in \PL'_\ell$ be intervals such that $[a \dd b]=[a' \dd b'] \cap [a'' \dd b'']$.

Assume first that $T[a \dd a+2\ell)$ is a p-square.
As $\PR_\ell$ was computed from uniform $\sigma$-runs, we have $\TR[a'+2\ell \dd b'+2\ell)=\TR[a'+\ell \dd b'+\ell)$, so $\TR[a+2\ell \dd i+2\ell)=\TR[a+\ell \dd i+\ell)$.  By \cref{lem:comb1} applied for each subsequent position in $[a \dd i)$, $T[i \dd i+2\ell)$ is a p-square.

Now assume that $T[i \dd i+2\ell)$ is a p-square.
As $[a'' \dd b''] \in \PL'_\ell$, $[n-b''-2\ell \dd n-a''-2\ell] \in \PL_\ell$. By definition, we have $(\TL)^R[n-b'' \dd n-a'')=(\TL)^R[n-b''-\ell\dd n-a''-\ell)$, i.e., $\TL[a'' \dd b'')=\TL[a''+\ell \dd b''+\ell)$. In particular, $\TL[a \dd i)=\TL[a+\ell \dd i+\ell)$. By \cref{lem:comb1} applied for each position in $[a+1 \dd i]$ in decreasing order, $T[a \dd a+2\ell)$ is a p-square.
\end{proof}

By \cref{cor:CH,lem:EE_enc,lem:EE_enc2}, after $\Oh(n\sigma)$ preprocessing we can check if a given even-length substring of $T$ is a p-square in $\Oh(1)$ time using an $\LCP^\approx$ query. For each $\ell \in [1 \dd \floor{n/2}]$ and each interval $[a \dd b] \in \R_\ell$, we check if $T[a \dd a+2\ell)$ is a p-square. By \cref{clm:2}, if so, we obtain an interval of occurrences of p-squares, and if not, then none of the positions in $[a \dd b]$ is the start of a p-square of length $2\ell$ in $T$. The total number of intervals in $\R_\ell$ is linear in the number of uniform $\sigma$-runs in $\TR$ and $\TL$, so we obtain $\Oh(r)$ intervals of occurrences of p-squares. By \cref{clm:1}, we do not miss any p-squares.

The final step of reporting non-equivalent p-squares mimics an analogous algorithm for standard squares of Bannai, Inenaga, and Köppl~\cite{DBLP:conf/cpm/BannaiIK17}. We report non-equivalent p-squares at their leftmost occurrence in $T$ and use the $\LPF^\approx$ array to identify them in intervals of occurrences. More precisely, we compute the $\LPF^\approx$ array in $\Oh(n\sigma)$ time and construct in $\Oh(n)$ time a data structure for answering Range Minimum Queries (RmQs) over $\LPF^\approx$ in $\Oh(1)$ time per query; see \cite{DBLP:conf/latin/BenderF00,DBLP:journals/siamcomp/HarelT84}. For each of the $\Oh(r)$ intervals $[a \dd b]$ of starting positions of p-squares of length $2\ell$, we want to list all positions $i \in [a \dd b]$ such that $\LPF^\approx[i]<2\ell$ and report corresponding leftmost occurrences of p-squares $T[i \dd i+2\ell)$. We ask an RmQ on $\LPF^\approx$ on the interval $[a \dd b]$; let the minimum be attained for $i \in [a \dd b]$. If $\LPF^\approx[i] \ge 2\ell$, the algorithm is finished. Otherwise, we report $T[i \dd i+2\ell)$ and run the algorithm recursively on intervals $[a \dd i-1]$ and $[i+1 \dd b]$. The total time complexity is proportional to the number of reported p-squares plus 1. By~\cite{DBLP:conf/spire/HamaiTNIB24}, $T$ contains $\Oh(n\sigma)$ non-equivalent p-squares. This completes an $\Oh(n\sigma+r)$-time reduction to computing uniform $\sigma$-runs.

When reporting all p-squares that are distinct as strings, it suffices to replace the $\LPF^\approx$ array with the standard $\LPF=\LPF^=$ array. Such an array can be computed in $\Oh(n)$ time after the letters of the string have been sorted~\cite{DBLP:journals/ipl/CrochemoreI08}. Then the computations take $\Oh(n\sigma+r+\mathsf{output})$ time.
\end{proof}

By \cref{thm:KK++uni}, string $\TR$ (and $\TL$) contains $\Oh(n \sigma \log \sigma)$ uniform $\sigma$-runs. They can be computed in $\Oh(n \sigma \log \sigma)$ time; we use the algorithm of Kolpakov and Kucherov~\cite{DBLP:conf/esa/KolpakovK01,DBLP:journals/tcs/KolpakovK03} to compute all $\sigma$-runs in $\TR$ and then partition each $\sigma$-run to uniform $\sigma$-runs using kangaroo jumps. That is, let $\TR[a \dd b)$ be a $\sigma$-run with period. We compute two values:
\begin{align*}
    d_1 &= 1+\LCP(\TR[a \dd b-2\ell), \TR[a+\ell \dd b-\ell))\\
    d_2 &= 1+\LCP(\TR[a+\ell \dd b-\ell), \TR[a+2\ell \dd b))
\end{align*}
that, intuitively, find the first $\ell$-mismatching position in the $\sigma$-run, if any ($d_1$), and the first $\ell$-mismatching position after position $a+\ell-1$, if any ($d_2$). Let $d=\min(d_1,d_2)$. We then report a uniform $\sigma$-run $\TR[a \dd a+d+2\ell)$ and continue processing the (non-maximal) $\sigma$-run $\TR[a+d \dd b)$ until its length drops below $2\ell$. The $\LCP$-queries in $(\TR)^R$ are answered in $\Oh(1)$ time~\cite{DBLP:conf/latin/BenderF00,DBLP:journals/siamcomp/HarelT84}. Thus \cref{lem:red} implies \cref{thm:main2}.

\section{Counting Generalised Squares in \texorpdfstring{$\Oh(n \log n)$}{O(n log n)} Time}\label{sec:1}
In this section we show an algorithm that counts non-equivalent $\approx$-squares, for an SCER $\approx$ with encoding function $\E$, in $\Oh(\pi^\E(n)+n\rho^\E(n)+n\log n)$ time. 

For a string $T$ of length $n$ and positive integer $p \le n/2$, we denote
\[\Squares_p^\approx=\{i \in [1 \dd n-2p+1]\,:\,T[i \dd i+2p)\text{ is an $\approx$-square}\}.\]
An \emph{interval representation} of a set $X$ of integers is
$X=[i_1\dd j_1] \cup [i_2\dd j_2] \cup \dots \cup [i_t\dd j_t]$,
where $j_1+1 < i_2$, \ldots, $j_{t-1}+1 < i_t$; $t$ is called the \emph{size} of the representation.

Counterparts of the following lemma corresponding to order-preserving squares and Cartesian-tree squares were shown in \cite{DBLP:journals/iandc/GourdelKRRSW20} and \cite{DBLP:journals/tcs/ParkBALP20}, respectively. Our proof generalises these proofs; it can be found in \cref{app:sqp}.

\begin{restatable}{lemma}{lemsqp}\label{lem:esqupe}
Let $\approx$ be an SCER and $\E$ be its encoding function. Given a string $T$ of length $n$, the interval representations of the sets $\Squares_p^\approx$ for all $1 \le p \le n/2$ have total size $\Oh(n \log n)$ and can be computed in $\Oh(\pi^\E(n) +n\rho^\E(n)+n \log n)$ time.
\end{restatable}

\newcommand{\RangeCount}{\mathsf{RangeCount}}
We count non-equivalent $\approx$-squares using the $\LPF^\approx$ array. We would like to count an $\approx$-square at the position of its leftmost occurrence. 
For an integer array $A[1 \dd n]$, we denote a range count query for $i,j \in [1 \dd n]$, $x \in \mathbb{Z}$ as 
$\RangeCount_A(i,j,x)\,=\,|\{k \in [i \dd j]\,:\,A[k]<x\}|$.
Then our problem reduces to computing the sum
\begin{equation}\label{eq:sumRC}
\sum_{p=1}^{\floor{n/2}}\sum_{[i \dd j] \in \Squares_p} \RangeCount_{\LPF^\approx}(i,j,2p).
\end{equation}

We say that an array $A[1 \dd n]$ is a \emph{linear oscillation array} if $\sum_{i=1}^{n-1} |A[i+1]-A[i]| = \Oh(n)$. The following fact is folklore for the $\LPF$ array; for completeness, we prove it for $\LPF^\approx$.

\begin{lemma}\label{lem:loLPF}
For any SCER $\approx$, $\LPF^\approx$ is a linear oscillation array.
\end{lemma}
\begin{proof}
Let us note that for every $i \in [1 \dd n)$, we have
\begin{equation}\label{eq:ip1}
\LPF^{\approx}[i+1] \ge \LPF^{\approx}[i]-1.
\end{equation}
Indeed, let $\ell=\LPF^{\approx}[i]$. If $\ell \le 1$, \cref{eq:ip1} is trivial. If $\ell>1$, there exists an index $j \in [1 \dd i)$ such that $T[j \dd j+\ell)\approx T[i \dd i+\ell)$. As $\approx$ is an SCER, we have $T[j+1 \dd j+\ell) \approx T[i+1 \dd i+\ell)$, which proves that $\LPF^{\approx}[i+1] \ge \ell-1$.

We will show that every array $A[1 \dd n]$ that satisfies $A[i+1] \ge A[i]-1$ for all $i \in [1 \dd n)$ and has values in $[0 \dd n)$ is a linear oscillation array.
Let us denote
\[g=\sum_{i=1}^{n-1} \max(A[i+1]-A[i],0),\quad s=\sum_{i=1}^{n-1} \max(A[i]-A[i+1],0).\]
By $A[i] \le A[i+1]+1$, we have $s < n$. Further,
\[A[n]-A[1]=\sum_{i=1}^{n-1} (A[i+1]-A[i]) = g-s.\]
We have $A[n]-A[1]<n$, so $g<s+n<2n$. Therefore, $\sum_{i=1}^{n-1} |A[i+1]-A[i]|=g+s < 3n$, which concludes the proof.
\end{proof}

We will now show that the sum from \cref{eq:sumRC} can be computed in $\Oh(n \log n)$ time. We use a very simple abstract problem.

\defdsproblem{Counting Problem}{
A set $Y \subseteq [1 \dd n]$, initially empty, and an integer $\ell$, initially $\ell=0$.
}{One of: (1) insert an element $x \in [1 \dd n]$ to $Z$; (2) delete an element $x \in Y$ from $Y$; (3) increment $\ell$ by 1; (4) decrement $\ell$ by 1. After each operation, output $|Y \cap [\ell+1 \dd n]|$.}

The Counting Problem can be solved with a basic indicator array for $Y$.

\begin{lemma}\label{lem:ACP}
After $\Oh(n)$-time preprocessing, each operation in the Counting Problem can be answered in $\Oh(1)$ time.
\end{lemma}
\begin{proof}
We store an indicator array of bits $C[1 \dd n]$ such that $C[i]=1$ if and only if $i \in Y$. Initially, $C \equiv 0$.
We also store a value $a=|Y \cap [\ell+1 \dd n]|$; initially $a=0$. After each operation, the current value of $a$ is returned.

When we insert an element $x \in [1 \dd n]$ to $Y$, we set $C[x]$ to 1 and increment $a$ if $x>\ell$. 
Symmetrically, when we delete $x \in Y$ from $Y$, we set $C[x]$ to 0 and decrement $a$ if $x<\ell$. When we increment $\ell$, we subtract $C[\ell]$ from $a$. When we decrement $\ell$, we add $C[\ell]$ to $a$.
\end{proof}

We are ready to obtain the final main result.

\begin{proof}[Proof of \cref{thm:main1}]
First we show how to compute the number of non-equivalent $\approx$-squares in a length-$n$ string. The $\LPF^\approx$ array can be computed in $\Oh(\pi^\E(n)+n\rho^\E(n)+n \log^2 \log n/\log \log \log n)$ time (\cref{cor:CH}).
We compute the interval representations of the sets $\Squares_p^\approx$ using \cref{lem:esqupe} in $\Oh(\pi^\E(n)+n\rho^\E(n)+n \log n)$ time.
The interval representations have total size $\Oh(n \log n)$.
We need to compute the sum of results of $\Oh(n \log n)$ $\RangeCount$ queries as in \cref{eq:sumRC}. We will do it using the Counting Problem.

Let us bucket sort all start and end points of intervals across all sets $\Squares_p^\approx$ in $\Oh(n \log n)$ total time. We create an instance of the Counting Problem. For each position $k$ from 1 to $n$, we proceed as follows. First, for each interval $[k \dd j] \in \Squares_p^\approx$, for any $p$, we insert $2p$ to the set $Y$. Then, we change the current value $\ell$ in the problem to $\LPF^\approx[k]$ by performing increments or decrements, as appropriate. Afterwards, we add the returned value of the Counting Problem to the final result. Finally, for each interval $[i \dd k] \in \Squares_p^\approx$, for any $p$, we delete $2p$ from the set $Y$.

For each $p \in [1 \dd \floor{n/2}]$, the intervals in $\Squares_p$ are pairwise disjoint. By \cref{lem:esqupe}, $\Oh(n \log n)$ insertions and deletions are performed in the Counting Problem. The total number of increments and decrements is bounded by $\Oh(n)$ by \cref{lem:loLPF} (and the fact that the values in this array are in $[0 \dd n)$). In total, the sum \eqref{eq:sumRC} is computed in $\Oh(n \log n)$ time. We obtain the first part of \cref{thm:main1}.

As before, if one wishes to compute all $\approx$-squares that are distinct as substrings, it suffices to replace the $\LPF^\approx$ array in the algorithm by the standard $\LPF=\LPF^=$ array.
We obtain the second part of \cref{thm:main1}.
\end{proof}

\bibliographystyle{plainurl}
\bibliography{references} 

\appendix

\section{Proof of \texorpdfstring{\cref{lem:esqupe}}{Lemma 25}}\label{app:sqp}
Let $\approx$ be an SCER.
We say that an $\approx$-square $T[i\dd i+2p)$ is \emph{right shiftable} if
$T[i+1\dd i+2p]$ is an $\approx$-square and \emph{right non-shiftable} otherwise.
Similarly, we say that the $\approx$-square is \emph{left shiftable} if $T[i-1\dd i+2p-2]$ is an $\approx$-square
and \emph{left non-shiftable} otherwise.

An $\approx$-square $T[i\dd i+2p)$ is called \emph{right extendible} if $T[i\dd i+p] \approx T[i+p\dd i+2p]$ and \emph{right non-extendible} otherwise.
Similarly, the $\approx$-square is called \emph{left extendible} if $T[i-1\dd i+p) \approx T[i+p-1\dd i+2p)$ and \emph{left non-extendible} otherwise.
See \cref{fig:ext}.
The fact that $\approx$ is an SCER implies the following.

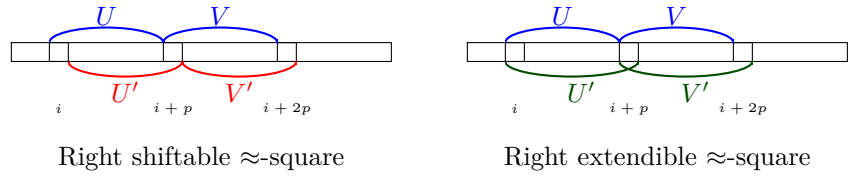
\begin{figure}[htpb]
\centering
\input{_fig_shiftable_squares}
\caption{Right shiftable square and right extendible square (if $U\approx V$ and $U'\approx V'$).}\label{fig:ext}
\end{figure}

\begin{observation}
A right (left) non-shiftable $\approx$-square is also right (left) non-extendible.
\end{observation}
Note that the converse is not necessarily true, i.e.\ some right (left) non-extendible
$\approx$-squares might be right (left) shiftable.

We start with the following lemma (based on \cite[Lemma 18]{DBLP:journals/tcs/CrochemoreIKKLP16}).

\begin{lemma}\label{lem:nonext_lcp}
Let $T$ be a string. Then $T[i\dd i+2p)$ is a non-extendible $\approx$-square if and only if $\LCP^\approx(T[i \dd n],T[i+p \dd n])=p$.
\end{lemma}
\begin{proof}
Let $\E$ be an encoding function for $\approx$.

($\Rightarrow$) If $T[i\dd i+2p)$ is a non-extendible 
$\approx$-square, then the longest common prefix of $\Code^\E(T[i\dd i+p])$ and $\Code^\E(T[i+p\dd i+2p])$ is exactly $\Code^\E(T[i\dd i+p))$. This yields that indeed $\LCP^\approx(T[i \dd n],T[i+p \dd n])=p$.

($\Leftarrow$) If $\LCP^\approx(T[i \dd n],T[i+p \dd n])=p$, then 
\[\Code^\E(T[i\dd i+p))=\Code^\E(T[i+p\dd i+2p))\text{ but }\Code^\E(T[i \dd i+p])\not=\Code^\E(T[i+p \dd i+2p]).\]
Hence, indeed $T[i \dd i+2p)$ is a non-extendible $\approx$-square.
\end{proof}

All non-extendible $\approx$-squares can be located
using the same algorithm as in computing branching tandem repeats (\cite[Theorem 1]{DBLP:journals/tcs/StoyeG02}), encapsulated in the following \cref{lem:Stoye}. We give a proof of the lemma for completeness.

A \emph{weighted tree} is a rooted tree with non-negative integer weights on nodes such that the weight of a node is smaller than the weight of its parent. We consider only trees in which each non-root and non-leaf node is branching.

\newcommand{\T}{\mathcal{T}}
\newcommand{\LL}{\mathit{LL}}
\begin{lemma}\label{lem:Stoye}
Let $\T$ be a weighted tree with $n$ leaves that are labeled by distinct numbers in $[1 \dd n]$. All pairs $(i,j)$ of leaf labels in $\T$, $1 \le i<j \le n$, such that the weight of their lowest common ancestor is $j-i$, can be listed in $\Oh(n \log n)$ time.
\end{lemma}
\begin{proof}
All leaves of $\T$ are stored in a global leaf list $LL$, left-to-right. Each node $v$ stores the sublist of $\LL$ that corresponds to the leaves in its subtree, denoted as $\LL(v)$, represented by the ends of the sublist, as well as the size $|\LL(v)|$ of the sublist. For each number in $[1 \dd n]$, the element of $\LL$ with this label is stored. This data can be computed in $\Oh(n)$ time via depth-first search.

For each node of $\T$, we store its pre-order and post-order numbers from the depth-first search. Then we can check if a given leaf $\ell$ belongs to a subtree of a given node $v$ in $\Oh(1)$ time (by checking inclusion of intervals).

In the main part of the algorithm, for each node $v$ of $\T$ with weight $w$, we compute all pairs of leaf labels $(i,j)$ such that $j-i=w$ and the corresponding leaves are located in different subtrees of $v$. To this end, let $u_1,\ldots,u_b$ be all children of $v$ and assume that $|\LL(u_1)|=\max_{t \in [1 \dd b]}\{|\LL(u_t)|\}$. For each leaf $i \in \LL(u_t)$, for $t > 1$, for each $j \in \{i-w, i+w\} \cap [1 \dd n]$, we check if the leaf with number $j$ is in the subtree of $v$ but not in the subtree of $u_t$; if so, we list the pair $(\min(i,j),\max(i,j))$.

The complexity is $\Oh(n \log n)$ by a known argument on small-to-large merging~\cite{DBLP:journals/jacm/Tarjan75}.
\end{proof}

\begin{lemma}
\label{lem:non-extendible-approx-squares}
All the (left and right) non-extendible $\approx$-squares in a string of length $n$ can be computed in $\Oh(\pi^\E(n) +n\rho^\E(n)+n \log n)$ time (here $\E$ is an encoding function for $\approx$).
\end{lemma}
\begin{proof}
We show the algorithm for right non-extendible $\approx$-squares; the computations for left non-extendible op-squares are symmetric. Let $T$ be a string of length $n$. We construct the $\approx$-suffix tree of $T$ in $\Oh(\pi^\E(n) +n\rho^\E(n))$ time (\cref{cor:CH}).

By \cref{lem:nonext_lcp}, $T[i \dd i+2p)$ is a non-extendible $\approx$-square if and only if the node $v$ with string label $\Code^\E(T[i \dd p))$ in the $\approx$-suffix tree is branching and the lowest common ancestor of leaves with string labels $\Code^\E(T[i \dd n])$ and $\Code^\E(T[i+p \dd n])$ is $v$. The algorithm behind \cref{lem:Stoye} allows us to report all such pairs $(i,i+p)$ in $\Oh(n \log n)$ time.
\end{proof}

\begin{lemma}\label{lem:nshift}
All the (left and right) non-shiftable $\approx$-squares in a string of length $n$ can be computed in $\Oh(\pi^\E(n) +n\rho^\E(n)+n \log n)$ time (here $\E$ is an encoding function for $\approx$).
\end{lemma}
\begin{proof}
We show the algorithm for right non-shiftable $\approx$-squares;
the computations for left non-shiftable $\approx$-squares are
symmetric.
Let $T$ be a string of length $n$. We use \cref{lem:non-extendible-approx-squares} for $T$. Now we need to filter out the non-extendible $\approx$-squares that are right shiftable.
For this, for a right non-extendible $\approx$-square $T[i\dd i +2p)$ we need to check if $\LCP^\approx(T[i + 1 \dd n], T[i + p + 1 \dd n]) = p-1$.
This condition can be verified in $\Oh(1)$ time  using \cref{cor:CH}.
\end{proof}



The only remaining part is the computation of an interval representation of the $\Squares_p^\approx$, but this can be done using left/right-non-shiftable squares in $\Oh(n\log n)$ time
due to \cite[Lemma 4.4]{DBLP:journals/iandc/GourdelKRRSW20}.
For the completeness of the description we briefly recall this procedure.

Let us define the following two auxiliary sets:
\begin{align*}{\cal L}_p &= \{ i : T[i\dd i+2p)\ \textrm{is a left non-shiftable $\approx$-square}\},\\
{\cal R}_p &= \{ i : T[i \dd i+2p)\ \textrm{is a right non-shiftable $\approx$-square}\}.
\end{align*}
By \cref{lem:nshift}, all the sets ${\cal L}_p$ and ${\cal R}_p$ can be computed in $\Oh(n\log n)$ time;
also $\sum_p|{\cal L}_p|=\Oh(n\log n)$.

If we present the set $\Squares_p^\approx$ for any $p$ as a sum of maximal intervals, then the left and right endpoints of these intervals will be ${\cal L}_p$ and ${\cal R}_p$, respectively.
Hence $|{\cal L}_p|=|{\cal R}_p|$. Thus if
\[{\cal L}_p=\{\ell_1,\ldots,\ell_k\}\ \textrm{and}\ {\cal R}_p=\{r_1,\ldots,r_k\},\]
then the interval representation of the set $\Squares_p^\approx$
is $[\ell_1\dd r_1]\cup \cdots \cup [\ell_k\dd r_k]$.
Clearly, it can be computed in $\Oh(|{\cal L}_p|)$ time.
We obtain \cref{lem:esqupe}. 

\end{document}

%% file: _fig_kruns.tex
\begin{tikzpicture}[xscale=0.3]
    \foreach \x/\a/\c in {
    1/a/black,
    2/b/black,
    3/a/black,
    4/c/black,
    5/a/black,
    6/a/black,
    7/b/black,
    8/a/black,
    9/a/black,
    10/b/black,
    11/a/black,
    12/b/black,
    13/a/black,
    14/a/black,
    15/c/black,
    16/a/black,
    17/a/black,
    18/b/black,
    19/c/black,
    20/b/black,
    21/a/black,
    22/a/black,
    23/b/black,
    24/a/black,
    25/c/black,
    26/a/black}{
        \draw (\x,0) node[above] {\textcolor{\c}{\texttt{\a}}};
    }
    \begin{scope}[yshift=-0.9cm]
    \foreach \x/\a/\c in {
    1/a/black,
    2/b/black,
    3/a/black,
    4/c/red,
    5/a/black,
    6/a/black,
    7/b/red,
    8/a/black,
    9/a/black,
    10/b/black,
    11/a/black,
    12/b/red,
    13/a/black,
    14/a/black,
    15/c/red,
    16/a/black}{
        \draw (\x,0) node[above] {\textcolor{\c}{\texttt{\a}}};
    }
    \draw (0.5,0.2) -- (0.5,0) -- (16.5,0) -- (16.5,0.2);
    \draw (8.5,0) -- (8.5,0.2);
    \end{scope}

    \begin{scope}[yshift=-1.4cm]
    \foreach \x/\a/\c in {
    2/b/black,
    3/a/black,
    4/c/red,
    5/a/black,
    6/a/black,
    7/b/red,
    8/a/black,
    9/a/black,
    10/b/black,
    11/a/black,
    12/b/red,
    13/a/black,
    14/a/black,
    15/c/red,
    16/a/black,
    17/a/black}{
        \draw (\x,0) node[above] {\textcolor{\c}{\texttt{\a}}};
    }
    \draw[xshift=1cm] (0.5,0.2) -- (0.5,0) -- (16.5,0) -- (16.5,0.2);
    \draw[xshift=1cm] (8.5,0) -- (8.5,0.2);
    \end{scope}

    \begin{scope}[yshift=-1.9cm]
    \foreach \x/\a/\c in {
    3/a/black,
    4/c/red,
    5/a/black,
    6/a/black,
    7/b/red,
    8/a/black,
    9/a/black,
    10/b/black,
    11/a/black,
    12/b/red,
    13/a/black,
    14/a/black,
    15/c/red,
    16/a/black,
    17/a/black,
    18/b/black}{
        \draw (\x,0) node[above] {\textcolor{\c}{\texttt{\a}}};
    }
    \draw[xshift=2cm] (0.5,0.2) -- (0.5,0) -- (16.5,0) -- (16.5,0.2);
    \draw[xshift=2cm] (8.5,0) -- (8.5,0.2);
    \end{scope}

    \begin{scope}[yshift=0.5cm]
    \foreach \x/\a/\c in {
    5/a/black,
    6/a/black,
    7/b/red,
    8/a/black,
    9/a/black,
    10/b/black,
    11/a/red,
    12/b/black,
    13/a/black,
    14/a/black,
    15/c/red,
    16/a/black,
    17/a/black,
    18/b/black,
    19/c/red,
    20/b/black}{
        \draw (\x,0) node[above] {\textcolor{\c}{\texttt{\a}}};
    }
    \draw[xshift=4cm] (0.5,0.2) -- (0.5,0) -- (16.5,0) -- (16.5,0.2);
    \draw[xshift=4cm] (8.5,0) -- (8.5,0.2);
    \end{scope}

    \begin{scope}[yshift=1cm]
    \foreach \x/\a/\c in {
    6/a/black,
    7/b/red,
    8/a/black,
    9/a/black,
    10/b/black,
    11/a/red,
    12/b/black,
    13/a/black,
    14/a/black,
    15/c/red,
    16/a/black,
    17/a/black,
    18/b/black,
    19/c/red,
    20/b/black,
    21/a/black}{
        \draw (\x,0) node[above] {\textcolor{\c}{\texttt{\a}}};
    }
    \draw[xshift=5cm] (0.5,0.2) -- (0.5,0) -- (16.5,0) -- (16.5,0.2);
    \draw[xshift=5cm] (8.5,0) -- (8.5,0.2);
    \end{scope}
    
    \begin{scope}[yshift=1.5cm]
    \foreach \x/\a/\c in {
    7/b/red,
    8/a/black,
    9/a/black,
    10/b/black,
    11/a/red,
    12/b/black,
    13/a/black,
    14/a/black,
    15/c/red,
    16/a/black,
    17/a/black,
    18/b/black,
    19/c/red,
    20/b/black,
    21/a/black,
    22/a/black}{
        \draw (\x,0) node[above] {\textcolor{\c}{\texttt{\a}}};
    }
    \draw[xshift=6cm] (0.5,0.2) -- (0.5,0) -- (16.5,0) -- (16.5,0.2);
    \draw[xshift=6cm] (8.5,0) -- (8.5,0.2);
    \end{scope}

    \begin{scope}[yshift=2cm]
    \foreach \x/\a/\c in {
    8/a/black,
    9/a/black,
    10/b/black,
    11/a/red,
    12/b/black,
    13/a/black,
    14/a/black,
    15/c/red,
    16/a/black,
    17/a/black,
    18/b/black,
    19/c/red,
    20/b/black,
    21/a/black,
    22/a/black,
    23/b/red}{
        \draw (\x,0) node[above] {\textcolor{\c}{\texttt{\a}}};
    }
    \draw[xshift=7cm] (0.5,0.2) -- (0.5,0) -- (16.5,0) -- (16.5,0.2);
    \draw[xshift=7cm] (8.5,0) -- (8.5,0.2);
    \end{scope}

    \begin{scope}[yshift=2.5cm]
    \foreach \x/\a/\c in {
    9/a/black,
    10/b/black,
    11/a/red,
    12/b/black,
    13/a/black,
    14/a/black,
    15/c/red,
    16/a/black,
    17/a/black,
    18/b/black,
    19/c/red,
    20/b/black,
    21/a/black,
    22/a/black,
    23/b/red,
    24/a/black}{
        \draw (\x,0) node[above] {\textcolor{\c}{\texttt{\a}}};
    }
    \draw[xshift=8cm] (0.5,0.2) -- (0.5,0) -- (16.5,0) -- (16.5,0.2);
    \draw[xshift=8cm] (8.5,0) -- (8.5,0.2);
    \end{scope}
    \foreach \x in {1,...,26}{
    \draw (\x,-0.3) node[above] {\tiny \x};
    }

\end{tikzpicture}

%% file: _fig_kruns_gr.tex
\begin{tikzpicture}[xscale=0.3]
    \foreach \x/\a/\c in {
    1/a/black,
    2/b/black,
    3/a/black,
    4/c/red,
    5/a/black,
    6/a/black,
    7/b/red,
    8/a/black,
    9/a/black,
    10/b/black,
    11/a/red,
    12/b/black,
    13/a/black,
    14/a/black,
    15/c/red,
    16/a/black,
    17/a/red,
    18/b/red,
    19/c/black,
    20/b/black,
    21/a/black,
    22/a/black,
    23/b/black,
    24/a/black,
    25/c/black,
    26/a/black}{
        \draw (\x,0) node[above] {\textcolor{\c}{\texttt{\a}}};
    }
    \begin{scope}[yshift=0.5cm]
    \foreach \x/\a/\c in {
    1/a/black,
    2/b/black,
    3/a/black,
    4/c/red,
    5/a/black,
    6/a/black,
    7/b/red,
    8/a/black,
    9/a/black,
    10/b/black,
    11/a/black,
    12/b/black,
    13/a/black,
    14/a/black,
    15/c/black,
    16/a/black,
    17/a/black,
    18/b/black}{
        \draw (\x,0) node[above] {\textcolor{\c}{\texttt{\a}}};
    }
    \end{scope}
    
    \begin{scope}[yshift=1cm]
    \foreach \x/\a/\c in {
    5/a/black,
    6/a/black,
    7/b/red,
    8/a/black,
    9/a/black,
    10/b/black,
    11/a/red,
    12/b/black,
    13/a/black,
    14/a/black,
    15/c/black,
    16/a/black,
    17/a/black,
    18/b/black,
    19/c/black,
    20/b/black,
    21/a/black,
    22/a/black}{
        \draw (\x,0) node[above] {\textcolor{\c}{\texttt{\a}}};
    }
    \end{scope}

    \begin{scope}[yshift=1.5cm]
    \foreach \x/\a/\c in {
    8/a/black,
    9/a/black,
    10/b/black,
    11/a/red,
    12/b/black,
    13/a/black,
    14/a/black,
    15/c/red,
    16/a/black,
    17/a/black,
    18/b/black,
    19/c/black,
    20/b/black,
    21/a/black,
    22/a/black,
    23/b/black,
    24/a/black}{
        \draw (\x,0) node[above] {\textcolor{\c}{\texttt{\a}}};
    }
    \end{scope}
    \foreach \x in {1,...,26}{
    \draw (\x,-0.3) node[above] {\tiny \x};
    }
    \draw[very thick] (7.5,0) rectangle (10.5,0.4);
    \draw[very thick,xshift=8cm] (7.5,0) rectangle (10.5,0.4);
    \draw[very thick, densely dotted] (10.5,0) rectangle (15.5,0.4);

\end{tikzpicture}

%% file: _fig_kruns_run.tex
\begin{tikzpicture}[xscale=0.3]
    \foreach \x/\a/\c in {
    1/b/red,
    2/b/black,
    3/d/red,
    4/a/black,
    5/a/red,
    6/a/black,
    7/a/black,
    8/b/black,
    9/a/black,
    10/a/black,
    11/b/black,
    12/a/black,
    13/a/black,
    14/b/black,
    15/a/black,
    16/a/red,
    17/b/black,
    18/a/black,
    19/a/black,
    20/b/red,
    21/a/black,
    22/c/black,
    23/b/black,
    24/a/black,
    25/a/black,
    26/c/black}{
        \draw (\x,0) node[above] {\textcolor{\c}{\texttt{\a}}};
    }
    \begin{scope}[yshift=0.5cm]
    \foreach \x/\a/\c in {
    2/b/black,
    3/d/red,
    4/a/black,
    5/a/red,
    6/a/black,
    7/a/black,
    8/b/black,
    9/a/black,
    10/a/black,
    11/b/black,
    12/a/black,
    13/a/black,
    14/b/black}{
        \draw (\x,0) node[above] {\textcolor{\c}{\texttt{\a}}};
    }
    \end{scope}
    
    \begin{scope}[yshift=1cm]
    \foreach \x/\a/\c in {
    4/a/black,
    5/a/red,
    6/a/black,
    7/a/black,
    8/b/black,
    9/a/black,
    10/a/black,
    11/b/black,
    12/a/black,
    13/a/black,
    14/b/black,
    15/a/black,
    16/a/black}{
        \draw (\x,0) node[above] {\textcolor{\c}{\texttt{\a}}};
    }
    \end{scope}

    \begin{scope}[yshift=1.5cm]
    \foreach \x/\a/\c in {
    6/a/black,
    7/a/black,
    8/b/black,
    9/a/black,
    10/a/black,
    11/b/black,
    12/a/black,
    13/a/black,
    14/b/black,
    15/a/black,
    16/a/black,
    17/b/black,
    18/a/black,
    19/a/black,
    20/b/black,
    21/a/black}{
        \draw (\x,0) node[above] {\textcolor{\c}{\texttt{\a}}};
    }
    \end{scope}

    \begin{scope}[yshift=2cm]
    \foreach \x/\a/\c in {
    11/b/black,
    12/a/black,
    13/a/black,
    14/b/black,
    15/a/black,
    16/a/red,
    17/b/black,
    18/a/black,
    19/a/black,
    20/b/black,
    21/a/black,
    22/c/black,
    23/b/black,
    24/a/black,
    25/a/black}{
        \draw (\x,0) node[above] {\textcolor{\c}{\texttt{\a}}};
    }
    \end{scope}

    \begin{scope}[yshift=2.5cm]
    \foreach \x/\a/\c in {
    15/a/black,
    16/a/red,
    17/b/black,
    18/a/black,
    19/a/black,
    20/b/red,
    21/a/black,
    22/c/black,
    23/b/black,
    24/a/black,
    25/a/black,
    26/c/black}{
        \draw (\x,0) node[above] {\textcolor{\c}{\texttt{\a}}};
    }
    \end{scope}

    \foreach \x in {1,...,26}{
    \draw (\x,-0.3) node[above] {\tiny \x};
    }
    \draw[very thick] (5.5,0) rectangle (21.5,0.4);

\end{tikzpicture}

%% file: _fig_shiftable_squares.tex
\begin{tikzpicture}

\begin{scope}
\draw (0,0) rectangle (5, 0.25);
\draw (0.5,0) rectangle +(0.25,0.25) node[midway,yshift=-0.75cm] {\tiny $i$};
\draw (2,0) rectangle +(0.25,0.25) node[midway,yshift=-0.75cm] {\tiny $i+p$};
\draw (3.5,0) rectangle +(0.25,0.25) node[midway,yshift=-0.75cm] {\tiny $i+2p$};

\begin{scope}
    \clip(0.5,0.25) rectangle +(3cm,2.0cm);
    \draw[thick,decorate,blue, decoration={bumps, segment length=3cm, amplitude=0.2cm}]
        (0.5,0.25)--+(10,0);
    \node[blue] at (1.25,0.6) {$U$};
    \node[blue] at (2.75,0.6) {$V$};
\end{scope}

\begin{scope}
    \clip(0.75,0) rectangle +(3cm,-2.0cm);
    \draw[thick,decorate,red, decoration={bumps, mirror, segment length=3cm, amplitude=0.2cm}]
        (0.75,0)--+(10,0);
    \node[red] at (1.5,-0.4) {$U'$};
    \node[red] at (3.0,-0.4) {$V'$};
\end{scope}

\node at (2.5,-1) [below] {Right shiftable $\approx$-square};
\end{scope}


\begin{scope}[xshift=6cm]
\draw (0,0) rectangle (5, 0.25);
\draw (0.5,0) rectangle +(0.25,0.25) node[midway,yshift=-0.75cm] {\tiny $i$};
\draw (2,0) rectangle +(0.25,0.25) node[midway,yshift=-0.75cm] {\tiny $i+p$};
\draw (3.5,0) rectangle +(0.25,0.25) node[midway,yshift=-0.75cm] {\tiny $i+2p$};

\begin{scope}
    \clip(0.5,0.25) rectangle +(3cm,2.0cm);
    \draw[thick,decorate,blue, decoration={bumps, segment length=3cm, amplitude=0.2cm}]
        (0.5,0.25)--+(10,0);
    \node[blue] at (1.25,0.6) {$U$};
    \node[blue] at (2.75,0.6) {$V$};
\end{scope}

\begin{scope}
    \clip(0.5,0) rectangle +(1.75cm,-2.0cm);
    \draw[thick,decorate,green!30!black, decoration={bumps, mirror, segment length=3.5cm, amplitude=0.2cm}]
        (0.5,0)--+(10,0);
\end{scope}
\begin{scope}
    \clip(2.0,0) rectangle +(1.75cm,-2.0cm);
    \draw[thick,decorate,green!30!black, decoration={bumps, mirror, segment length=3.5cm, amplitude=0.2cm}]
        (2.0,0)--+(10,0);
\end{scope}

\node[green!30!black] at (1.5,-0.4) {$U'$};
\node[green!30!black] at (3.0,-0.4) {$V'$};
\node at (2.5,-1) [below] {Right extendible $\approx$-square};

\end{scope}

\end{tikzpicture}